\documentclass[conference,letterpaper]{IEEEtran}

\usepackage[utf8]{inputenc} 
\usepackage[T1]{fontenc}
\usepackage{url}
\usepackage{ifthen}
\usepackage{cite}
\usepackage[cmex10]{amsmath} 

\IEEEoverridecommandlockouts
\usepackage{amssymb,amsfonts}
\usepackage{algorithm}
\usepackage{graphicx}
\usepackage{textcomp}
\usepackage{xcolor}
\usepackage{algorithm}  
\usepackage{algpseudocode}  
\usepackage{amsmath} 
\usepackage{amssymb}
\usepackage{changepage}
\usepackage{setspace}
\usepackage{pdfpages}
\usepackage{xcolor}
\usepackage{amsthm} 
\usepackage{algorithm}
\usepackage{algorithm}

\usepackage{cite}
\usepackage{amssymb,amsfonts}
\usepackage{graphicx}
\usepackage{textcomp}

\usepackage{algorithm}
\newtheorem{theorem}{Theorem}
\newtheorem{remark}{Remark}
\newtheorem{Corollary}{Corollary}
\newtheorem{Definition}{Definition}
\newtheorem{Example}{Example}
\newtheorem{Lemma}{Lemma}
\usepackage{comment}
\usepackage{booktabs,color,amssymb,amsmath}

\begin{document}
	
	\title{Optimal Coding Scheme and Resource Allocation for  Distributed Computation with Limited Resources}
	
	\author{\IEEEauthorblockN{Shu-Jie Cao, Lihui Yi, Haoning Chen and Youlong Wu\\} 
\thanks{This work is supported   by   NSFC grant NSF61901267.
}	 
		\IEEEauthorblockA{
			ShanghaiTech University, Shanghai, China}
		\{caoshj, yilh, chenhn, wuyl1\}@shanghaitech.edu.cn
	}

	\maketitle
	
	\begin{abstract}
		A central issue of distributed computing systems is how to optimally allocate  computing and storage resources  and design  data shuffling strategies such that the total execution time for computing and data shuffling is minimized. This is extremely  critical when the computation, storage and communication resources are limited. In this paper, we study the  resource allocation and coding scheme for the MapReduce-type framework with \emph{limited} resources. In particular, we focus on the  coded  distributed computing (CDC) approach proposed by Li \emph{et al.}. {
					  We first extend the  asymmetric CDC (ACDC) scheme proposed by Yu \emph{et al.} to the cascade case where each output function is computed by multiple servers. Then we demonstrate that  whether CDC or ACDC is better depends on  system parameters (e.g., number of computing servers) and task parameters (e.g., number of input files), implying that neither CDC nor ACDC is optimal. By merging the ideas of CDC and ACDC, we  propose a hybrid scheme  and show that it can strictly outperform  CDC and ACDC.  
		Furthermore, we derive an information-theoretic converse showing that for the MapReduce task using a type of weakly symmetric Reduce assignment, which includes the Reduce assignments of CDC and ACDC as special cases,  the hybrid scheme with a corresponding resource allocation strategy is optimal, i.e., achieves the minimum execution time, for an arbitrary amount of computing servers and storage memories.  
		}

	\end{abstract}
	
	
	\section{Introduction}

	Distributed computing   has attracted significant interests as it enables  complex computing tasks to process in parallel across many computing nodes to speed up the computation.  However, due to massive data and limited communication   {resources},  the distributed computing systems  suffer from  {the} communication bottleneck \cite{ComBottleneck}. Many previous works have shown that  using coding  can greatly reduce communication load (see, e.g., \cite{b1,b2,Straggler_2,MR,Edge,CCEP,Scalable,ACM'20,Fu'Glb19,Wan'aXiv19,CompressedCDC,Leveraging,NewCombinatorial,CodedTerasort,MachineLearning}).  
	
	In \cite{b1}  Li \emph{et al.} considered  a  MapReduce-type framework   consisting of three phases: Map, Shuffle and Reduce, and   {proposed the coded distributed computing (CDC) scheme.  In the CDC scheme,  $K$ servers first map their  stored files into intermediate values in the Map phase, and then based on the mapped intermediate values,  the servers multicast  coded symbols   to other servers in the Shuffle phase, and finally, each server computes   output functions based on the local mapped intermediate values and the received coded symbols.  The CDC scheme was generalized to the \emph{cascaded} case where each Reduce function is computed by   $s \geq 1$ servers, which is helpful in  reducing the communication load of the next-round data shuffling  when the job consists of multiple rounds of  computations. Based on the CDC scheme of case $s=1$,   an asymmetric coded distributed {computing} (ACDC) scheme  was proposed  in \cite{b2} which 
	allows a set of  servers serving as ``helper'' to perform Map and Shuffle operations, but not Reduce operation. The ACDC scheme   achieves the minimum execution time when the total number of computing servers   and size of storage memories are sufficiently large.  This is, however, impractical in some real distributed systems which only have limited computing and storage resources.} 
	
	In this paper, we try to answer the following questions. Given a MapReduce-type task with an \emph{arbitrary} amount of computing servers and storage  memories, 1)  is it always good to use all available computing servers? If not, how many servers should be exactly used? 2) how to efficiently  utilize the storage memories and allocate files to servers?  3) how to efficiently exchange  information among computing servers? We answer  these questions by establishing an optimal coding scheme and resource allocation strategy.  
	{In more detail, we first show that  neither the CDC nor ACDC scheme is optimal, and whether the CDC or ACDC scheme is better depends on  system parameters (e.g., number of computing servers and size  of storage memories) and task parameters (e.g., number of input files).  Then, we propose a hybrid coding scheme for the case $s\geq1$ by combining the ideas of  CDC and ACDC, and show that this hybrid scheme  strictly outperforms CDC and ACDC. 
	The generalized ACDC scheme  for $s\geq 1$  is similar in spirit to  the  coded caching schemes  in \cite{Centralized,WanKai'19}.   
By deriving   an information-theoretic converse on the execution time, 
	 we prove that for any MapReduce task using a weakly symmetric Reduce assignment which includes  the assignments of CDC and ACDC as special cases,  our scheme achieves the minimum execution time.   The optimality result holds for {an} arbitrary amount of computing servers and storage memories. 
	 
	}
	




	\section{Problem Formulation}
	
Let $Q,N,F,B, K, M \in \mathbb{N}$ be some positive integers, and define notations   $\mathcal{Q}\triangleq\{1,\ldots,Q\}$, $\mathcal{K}\triangleq\{1,\ldots,K\}$ and $\mathcal{N
}\triangleq\{1,\ldots,N\}$.     Given $N$ input files $\omega_{1},\dots,\omega_{N} \in \mathbb{F}_{2^{F}}$,  a  task wishes  to compute $Q$ output functions $\phi_{1},\dots,\phi_{Q}$, where $\phi_{q}$: $(\mathbb{F}_{2^{F}})^{N} \rightarrow \mathbb{F}_{2^{B}}$, $q \in \{1,\dots,Q\}$, maps all input files into a ${B}$-bit output value $u_q = \phi_{q}(\omega_{1},\dots,\omega_{N}) \in \mathbb{F}_{2^{B}}$.   

We consider a MapReduce-type framework with \emph{limited} computing and storage resources, in which  there are  $K$ available computing servers, and all servers  in total can store up to  $M$ files ($M\geq N$, otherwise the task cannot be completed). In \cite{Ezzeldin'ITW17,Yan'ISIT19} the authors showed that the storage and computation cost can be reduced by letting each node    choose to calculate the intermediate values only if they are used subsequently. Here we focus on the storage cost of storing input files for simplicity. We believe that the extension of our work which jointly considers  the cost of  storing files, intermediate values and output values will not be  hard.

We call $(Q,N,F,B)$ the \emph{task parameters} as they are determined by the intrinsic features of the computation task, and  $(K,M)$ the \emph{system parameters} as they are related to  the available resources of the computation system.
	
	The computation of the output function $\phi_{q}$, $q \in \mathcal{Q}$ can be decomposed as
	$
	\phi_{q} (\omega_{1},\!\ldots\!,\omega_{N}\!) \!= \! h_{q}\!\left((g_{q,1}(\omega_{1}),\!\ldots\!,(g_{q,{N}}(\omega_{N}))\right)$,  where the ``Map'' function $g_{q,n}$ and ``Reduce'' function $h_{q}$, for $n \in\mathcal{N}$ and $q \in \mathcal{Q}$,  are illustrated as follows: 
 	\begin{itemize}
		\item $g_{q,n}:\,\mathbb{F}_{2^{F}} \rightarrow \mathbb{F}_{2^{V}}$   maps the input file $\omega_{n}$ into a length-$V$, $V \in \mathbb{N}$,  \emph{intermediate value} $v_{q,n} = g_{q,n}(\omega_{n}) \in \mathbb{F}_{2^{V}}$;
		\item $h_{q}:(\mathbb{F}_{2^{V}})^{N} \rightarrow \mathbb{F}_{2^{B}}$   maps the intermediate values of the output function $\phi_{q}$ in all input files into the output value $u_{q} = h_{q}(v_{q,1},\dots,v_{q,N}) = \phi_{q}(\omega_{1},\dots,\omega_{N})$.
	\end{itemize}
	

	The whole computation work proceeds in the following three phases: \emph{Map}, \emph{Shuffle} and \emph{Reduce}.
	\subsubsection{Map Phase} Denote the indices of files  mapped by  Node $k\in \mathcal{K}$ as $\mathcal{M}_{k}\subseteq \mathcal{N}$. Every file is mapped by some nodes (at least one), i.e., $ \bigcup_{k \in \mathcal{K}}\mathcal{M}_{k} =\mathcal{N}$. Node $k$ computes the Map function $({g}_{1,n}(\omega_{n}),\dots, {g}_{Q,n}(\omega_{n})) = (v_{1,n},\dots,v_{Q,n})$, if $n\in \mathcal{M}_{k}$. 
	
Since all nodes can only store up to $M$ input files, we have: 
	\begin{IEEEeqnarray}{rCl}\label{Storage}
	 \sum_{k \in\mathcal{K}} |\mathcal{M}_{k}|\leq M.
\end{IEEEeqnarray}

	\subsubsection{Shuffle Phase} 
	
	A message $\emph{X}_{k}\in\mathbb{F}_{2^{\ell_{k}}}$, for some $\ell_k \in \mathbb{N}$, is generated by each Node $k\in \mathcal{K}$, as a function of intermediate values computed locally, i.e., $\emph{X}_{k} = \psi_{k}(v_{1,n},\dots,v_{Q,n}:n \in \mathcal{M}_{k})$. Then Node $k$ multicasts it to  other nodes through a shared noiseless link.
	

	\subsubsection
	{Reduce Phase}   
	Assign $Q$ Reduce functions to $K$ nodes. Each output function will be  computed by $s$ nodes ($ s\in\mathcal{K}$). Denote $\mathcal{W}_{k}\subseteq \mathcal{Q}$ as the  assignment indices of Reduce functions on the Node $k \in \mathcal{K}$, with $ \bigcup_{k \in \mathcal{K}} \mathcal{W}_{k} = \mathcal{Q}.$ 	Each Node $k$ produces  output values $u_q = h_q(v_{q,1},\dots,v_{q,N})$ for all $q\in \mathcal{W}_{k}$, based on the local Map intermediate values $\{v_{1,n},\dots,v_{Q,n}:  n \in \mathcal{M}_{k}\}$ and the received messages $({X}_{k}:k \in\mathcal{K})$.

	Unlike the symmetric Reduce design in \cite{b1} where   $Q$ Reduce functions are assigned symmetrically to all $K$ nodes, we allow a more general  Reduce design where some nodes may not produce any Reduce function. This will cause more complexities both in the scheme design and converse proof, but is worth as  it can reduce the time cost in the Shuffle phase (see  \cite{b2} or  Section \ref{SecResults} ahead). 
	   We now introduce a weakly symmetric Reduce assignment, denoted by $\mathcal{W}_\text{symmetric}$ that is evolved from the Reduce design in \cite{b1}: 
\begin{Definition}{(Weakly Symmetric Reduce Assignment $\mathcal{W}_\textnormal{symmetric}$)}\label{DefFilePm}
Given a task parameter $s$ and a designed $\mathcal{K}_s$,  evenly split $Q$ Reduce functions into   $\binom{|\mathcal{K}_s|}{s}$ disjoint batches of size $\frac{Q}{\binom{|\mathcal{K}_s|}{s}}\in\mathbb{N}$\footnote{We focus on the case  ${Q}/{\binom{|\mathcal{K}_s|}{s}}\in\mathbb{N}$, like that in \cite{b1}.}, each corresponding to a subset $\mathcal{P}\subseteq \mathcal{K}_s$ of size $s$, i.e., 
\begin{IEEEeqnarray}{rCl}
\mathcal{Q} =\cup_{\mathcal{P}\subseteq\mathcal{K}_s:|\mathcal{P}|=s }\mathcal{D}_{\mathcal{P}},
\end{IEEEeqnarray}
where   $\mathcal{D}_{\mathcal{P}}$ denotes the batch of $\frac{Q}{\binom{|\mathcal{K}_s|}{s}}$ Reduce functions about the subset $\mathcal{P}$. Node $k\in\mathcal{K}_s$ computes the  Reduce functions whose indices are in $\mathcal{D}_{\mathcal{P}}$ if $k\in \mathcal{P}$. In this way,  	
	we have 
	\begin{IEEEeqnarray}{rCl}\label{eqReduce}
	\text{$|\mathcal{W}_i|=\frac{sQ}{K_s}\in\mathbb{N}$,  $\forall i\in \mathcal{K}_s$, and $|\mathcal{W}_j|=0$,  $\forall j\notin \mathcal{K}\backslash\mathcal{K}_s$}.
\end{IEEEeqnarray}
\end{Definition}


Note that here only the nodes in $\mathcal{K}_s$ produce Reduce functions, and when $\mathcal{K}_s=\mathcal{K}$, the Reduce design  $\mathcal{W}_\text{symmetric}$ turns to be the same as that in the CDC scheme. Also, when $s=1$, $\mathcal{W}_\text{symmetric}$  is simply to let each node in $\mathcal{K}_s$ produce $Q/K_s$ Reduce functions.

We focus on the case where the computations in Map and Reduce phase run in parallel,  while the Map, Shuffle, and Reduce  phases take place in a sequential manner.  We adopt the  following definitions the same as in    \cite{b2}:  
	\begin{Definition}[Peak Computation Load]
		The peak  computation load  is defined to be  $p \triangleq \frac{\max \limits_{k\in\{1,\dots,K\}}|\mathcal{M}_{k}|}{N}$.  \label{defCpLoad}
	\end{Definition}
	
	\begin{Definition}[Communication Load]
		The   communication load  is defined  as the total number of bits communicated by the $K$ nodes, i.e.,
		$L\triangleq\frac{\sum_{k\in\mathcal{K}}\ell_{k}}{NQV}$.  \label{defCnLoad}
	\end{Definition}

	\begin{Definition}[Execution Time] \label{DefTime}
	Denote the time consumed in Map, Shuffle and Reduce phase as $T_{\textnormal{map}}$, $T_{\textnormal{shuffle}}$ and $T_{\textnormal{reduce}} $, respectively. Define 
	\begin{IEEEeqnarray}{rCl}
	T_{\textnormal{map}} &\triangleq& \max \limits_{k\in\mathcal{K}} 
	c_\textnormal{m}\frac{|\mathcal{M}_{k}|}{N} =c_\textnormal{m}p,\\
	T_{\textnormal{shuffle}} &\triangleq&c_\textnormal{s}L,~ T_{\textnormal{reduce}} \triangleq   c_\textnormal{r} \max \limits_{k\in\mathcal{K}} ~|\mathcal{W}_{k}|,
\end{IEEEeqnarray}
for some system parameters $c_\textnormal{m}, c_\textnormal{s},c_\textnormal{r}>0$.

	Define the achievable execution time with parameter $s$ and using the Reduce assignment $\mathcal{W}_\textnormal{symmetric}$ as 
	\begin{IEEEeqnarray}{rCl}
	T_\mathcal{W}(s)\triangleq T_{\textnormal{map}} \!+\!T_{\textnormal{shuffle}} \!+\!T_{\textnormal{reduce}}.
\end{IEEEeqnarray} The minimum execution time is denoted by $T_\mathcal{W}^*(s)$.
	\end{Definition}
	{
	The goal is to design the  Map, Shuffle, Reduce operations and a resource allocation strategy such that the execution time of a given  MapReduce-type task with  system parameters $(K,M,c_\textnormal{m},c_\textnormal{s},c_\textnormal{r})$ and task parameters $(N,Q,s)$,  is minimized.}
	
	\section{Motivation and Examples}\label{moti}
	Consider a MapReduce-type task with     $s = 1$. If   the CDC scheme \cite{b1} is applied, then 
	the following execution time $T_\textnormal{CDC}$ is achievable: For $K_c \le K$ and  $r_1\leq M/N, r_1\in\{1,\ldots,K_c\}$,  
	\begin{IEEEeqnarray}{rCl}
		T_\textnormal{CDC}= c_\textnormal{m}\frac{r_1}{K_c} + c_\textnormal{s} \frac{1}{r_1 } \left( 1- \frac{r_1}{K_c}\right)+  c_\textnormal{r}\frac{Q}{K_c}.
	\end{IEEEeqnarray}

If   the ACDC scheme \cite{b2} is applied, then 
	the following execution time $T_\textnormal{ACDC}$ is achievable: For $K_s+K_h \le K$ and  $r_2\leq M/N, r_2\in\{1,\ldots,K_s\}$,  
	{\begin{IEEEeqnarray}{rCl}\label{EqACDC}
		T_\textnormal{ACDC}	=c_\textnormal{m}\max\left\{\!\frac{r_2}{K_s},\frac{1}{K_h}\!\right\} \!+\!  c_\textnormal{s}  \frac{1}{r_2\!+\!1 }\! \left( \!1\!-\! \frac{r_2}{K_s}\!\right) \!\!+ \! c_\textnormal{r}\frac{Q}{K_s}.~
	\end{IEEEeqnarray}}
	
	In \cite{b2}  it showed that if $K$ and $M$ are sufficiently large such that 
	$K\geq K'\geq Q$ and $M \geq M' = (r'_2 + 1)N$, where  
	\begin{subequations}
		\begin{align}
			& r'_2= \mathop{\arg\min}_{r_2\in\{0,\ldots,Q\}}   c_\textnormal{m} \frac{r_2}{Q}  \!   +\!  c_\textnormal{s}  \frac{1}{r_2\!+\!1 } \left( 1\!-\! \frac{r_2}{Q}\right), \\
			&K' =  \left\{\begin{array}{cc} Q+\lceil{{Q}/{r_2'}}\rceil,~&~\text{~if~}0<r'_2< Q  \\ Q,&~\text{~if~} r_2'=Q \end{array}, \right.
			\end{align}\label{ACDC optimal}
	\end{subequations}
then $T_\textnormal{ACDC}$ is the minimum execution time. 
	
	Unfortunately, the ACDC scheme is not optimal in   general.  In the following examples, we
show that whether CDC or ACDC  is better depends on the task and system parameters.
	
	\begin{Example}
		Consider a MapReduce-type task with $N=Q=4$, $K = 3$, $M = 12$ and $c_\textnormal{m} = 2c_\textnormal{s}$.
		
		When applying the CDC scheme with $K_c =3$ and   $r_1 = 1$, which  does not violate any resource constraint, we have $T_\textnormal{CDC} = \frac{4}{3}c_\textnormal{s} + \frac{4}{3}c_\textnormal{r}$. When applying the ACDC schme, there are only three possible allocations due to the constraints of $K$ and $M$. Here we list all of them 1) $(K_s,K_h,r_2) = (2,1,0)$, $T_\textnormal{ACDC} = 3c_\textnormal{s} + 2c_\textnormal{r}$; 2) $(K_s,K_h,r_2) = (2,1,1)$,  $T_\textnormal{ACDC} = \frac{9}{4}c_\textnormal{s} + 2c_\textnormal{r}$; 3) $(K_s,K_h,r_2) = (1,2,0)$, $T_\textnormal{ACDC} = 2c_\textnormal{s} + 4c_\textnormal{r}$.
		
		 {It can be seen that} $T_\textnormal{CDC} < T_\textnormal{ACDC}$ always holds, indicating that when the amount of resources is not large enough, the execution time of the CDC scheme may be shorter than that of the ACDC scheme.\label{eg1}
	\end{Example}
	%
	%
	%
	\begin{Example}
		Reconsider the example above, but with $K = 6$.  
		Assume   $c_\textnormal{r}$ is too small such that the Reduce time can be ignored. 
		For the CDC scheme, the optimal allocation is $K_c = 6,r_1 = 1$, and then $T_\textnormal{CDC} = \frac{7}{6}c_\textnormal{s}$. For the ACDC scheme, it can be checked that the best choice is $K_s = 3, K_h = 3,r_2 = 1$, and then $T_\textnormal{ACDC} = c_\textnormal{s}$.  	Therefore, $T_{\textnormal{ACDC}}<T_{\textnormal{CDC}}$, indicating that    when the resources are sufficiently large, the execution time of the ACDC scheme could be shorter than that of the CDC scheme.
	\end{Example}

In Fig. 1 we compare the execution time of the CDC and ACDC schemes, demonstrating that which scheme is better varies with the  number of  computing nodes $K$.
	
			\begin{figure}
			\centering
			\includegraphics[width=0.45\textwidth]{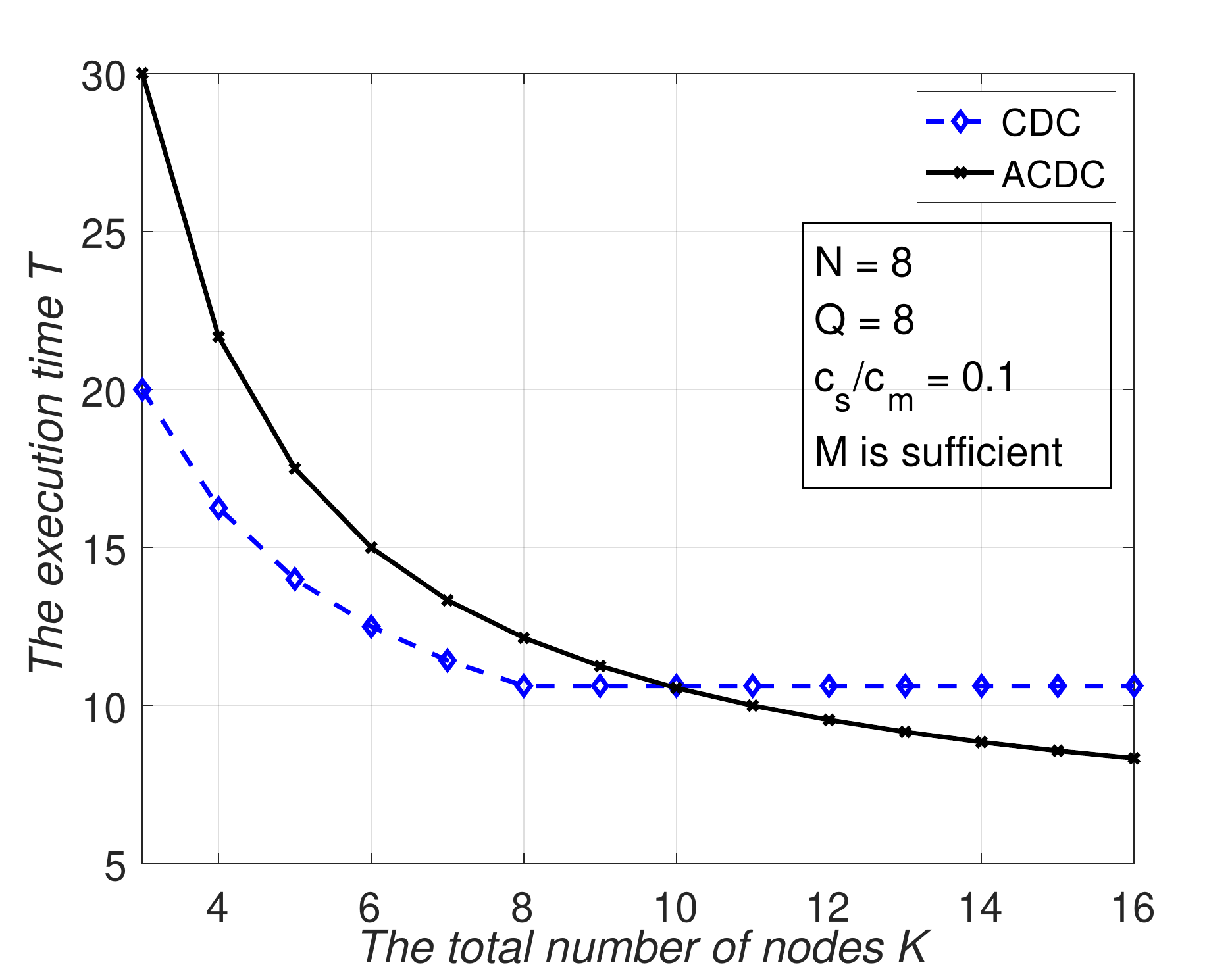} 
			\caption{Execution time for $s = 1$, neglecting the Reduce time.}
			\label{img} 
		\end{figure}
	
	%
	%

	\section{Main Result}\label{SecResults}
	
For some   nonnegative integers $r_1,r_2 \in \{0, \ldots, K\}$, let 
\begin{subequations} \label{EqL1} 
\begin{IEEEeqnarray}{rCl}
L_1(r_1,s,K)  &\triangleq & \mathop{\sum}\limits_{\ell = \max\{r_1+1,s\}}^{\min\{r_1+s,K\}} \frac{\binom{K}{\ell} \binom{\ell-1}{r_1} \binom{r_1}{\ell-s}}{ \binom{K}{r_1} \binom{K}{s}}  \frac{\ell}{\ell - 1}, ~~\label{LoadCDC}\\
L_2(r_2, s,K)  &\triangleq & \mathop{\sum}\limits_{\ell = \max\{r_2+1, s\}}^{\min\{r_2+s, K\}} \frac{\binom{K}{\ell} \binom{\ell - 1}{r_2} \binom{r_2}{\ell - s}}{\binom{K}{r_2} \binom{K}{s}}.~ \label{LoadACDC}
\end{IEEEeqnarray}
\end{subequations}
{
	\begin{theorem}\label{theorem T}
		For a MapReduce-type task with  system parameters $(K,M,c_\textnormal{m},c_\textnormal{s},c_\textnormal{r})$, task parameters $(N,Q,s)$ and a Reduce design $\mathcal{W}_\textnormal{symmetric}$,  the minimum execution time is 
		\begin{subequations}\label{general T} 
 \begin{IEEEeqnarray}{rCl}
	 T^*_\mathcal{W}(s) 
		=&& \min_{\substack{(\alpha, r_1,r_2,K_s,K_h)}} \!\!c_\textnormal{m}\cdot \max \left\{ \frac{\alpha r_1 \!+\!  \bar{\alpha}r_2}{K_s} ,\frac{\bar{\alpha}}{K_h}\right\}~~~~~\nonumber\\ 
		\!+&&\alpha  c_\textnormal{s} L^*_1(r_1,s,K_s)  \!+\! \bar{\alpha}c_\textnormal{s}  L^*_2(r_2,s,K_s)\!+\!c_\textnormal{r}\frac{sQ}{K_s},  ~~\\
		s.t. ~&&  K_s + K_h  \leq   K,~1 \leq K_s \leq Q,~~~~\label{eqConK}\\
		&&\alpha r_1 \!+\! \bar{\alpha}(r_2\!+\!1) \leq  \frac{M}{N},~0\leq \alpha\leq 1,\\
		&& r_1,r_2\in \{0,\ldots,K_s\},\label{eqConr1r2}
\end{IEEEeqnarray}
\end{subequations}
		 where  $\bar{\alpha}\triangleq 1-\alpha$, $L_1^*(r_1, s,K_s)$ is the lower convex envelope of the points $\{(r_1, L_1(r_1,s, K_s)): r_1\in \{0, \dots, K_s\}\}$, and  $L_2^*(r_2, s,K_s)$ is the lower convex envelope of the points $\{(r_2, L_2(r_2,s, K_s)): r_2\in \{0, \dots, K_s\}\}$. 
	\end{theorem}
	\begin{proof}  See the achievability proof in Section \ref{Scheme}, and converse proof in  Section \ref{converse}.
	\end{proof}


	Letting $s=1$ in Theorem  \ref{theorem T}, we obtain  the following corollary:
	\begin{Corollary}
		For a MapReduce-type task with  system parameters $(K,M,c_\textnormal{m},c_\textnormal{s},c_\textnormal{r})$, task parameters $(N,Q,s=1)$ and a Reduce design $\mathcal{W}_\textnormal{symmetric}$, the minimum execution time is
	    \begin{IEEEeqnarray}{rcl}\label{TimeMixed}
				\notag T^*_\mathcal{W}(1) =&& c_\text{m}\cdot \max \left\{ \frac{\alpha r_1 +  \bar{\alpha} r_2}{K_s} ,\frac{\bar{\alpha}}{K_h}\right\}\!+ \! c_\textnormal{r}\frac{Q}{K_s}\nonumber\\
				&&+    \alpha c_\textnormal{s}\frac{1}{r_1}\left(\!1\! -\! \frac{r_1}{K_s}\right)\! +\! \bar{\alpha} c_\textnormal{s}  \frac{1}{r_2\!+\!1 } \!\left(\! 1\!-\! \frac{r_2}{K_s}\!\right),
		\end{IEEEeqnarray}
for some $(r_1,r_2,K_s,K_h,\alpha)$ satisfying constraints in (\ref{eqConK}--\ref{eqConr1r2}).		
	\end{Corollary}
	
	\begin{remark}\label{Improvment}
		The achievable execution time in Theorem \ref{theorem T}  strictly outperforms that of CDC and ACDC. This can be easily seen in Fig. 1, where the curve of $\min\{T_\textnormal{CDC}, T_\textnormal{ACDC}\}$ is   not   lower convex  (not convex in the regime $8\leq K\leq 12$), while the achievable execution time in \eqref{general T} is a lower convex function due to the   ``time sharing''  parameter $\alpha$.  
	\end{remark}

	\begin{remark} 
		When $K$ and $M$ are sufficiently large such that $K \ge K',M \ge M'$ where $K',M'$ are given in \eqref{ACDC optimal},  the ACDC scheme is optimal. On the   {contrary},  when the amount of resources  is not sufficiently large, the execution time of CDC could be shorter than that of ACDC, e.g., when $K \le \frac{M}{N} \le Q + 1$ and $\frac{c_\textnormal{s}}{c_\textnormal{m}} \le \frac{2}{K}$. This is identical to the result shown in Example \ref{eg1} given in Section \ref{moti}.

	\end{remark}

\begin{remark}\label{CDC}
		{ For the case $s=1$ with deficient resources  such that the CDC scheme is optimal, i.e., $\alpha^*=1$, 
	  if ignoring the integrality constraints and  the Reduce time $T_\textnormal{reduce}$, we derive the optimal choice of $r^*_1$ and   $K^*_s$:  
		\begin{align*}
		(r_1^*\!, K_s^*)\!=\!\! \left\{ 
		\begin{aligned}
		&(1,\mathsf{K}), \quad c_\textnormal{s}/c_\textnormal{m} \leq   1/ \mathsf{K},\\
		&(\sqrt{c_\textnormal{s}/c_\textnormal{m}},\mathsf{K}), ~ \mathsf{K} \le M/N ,1/ \mathsf{K} \le c_\textnormal{s}/c_\textnormal{m} \le \mathsf{K},\\
		&  \quad \quad \quad \quad~~\textnormal{or } \mathsf{K} \!\!\ge\!\! M\!/\!N , 1\!/ \mathsf{K} \!\!\le\!\! c_\textnormal{s}/c_\textnormal{m} \!\!\le\!\! M^2\!/\!(N^2\mathsf{K}),\\
		&(M/N,\mathsf{K}),  ~~ \mathsf{K} \!\!\ge\!\! M\!/\!N, M^2\!/\!(N^2\mathsf{K}) \!\!\le\!\! c_\textnormal{s}/c_\textnormal{m} \!\!\le\!\! M\!/\!N,\\
		&(M/N,M/N),~~ \mathsf{K} \ge M/N,c_\textnormal{s}/c_\textnormal{m} \ge M/N,\\
		&(\mathsf{K},\mathsf{K}),~~ \mathsf{K} \le M/N,c_\textnormal{s}/{c_\textnormal{m}} \ge 1/{\mathsf{K}},
		\end{aligned}
		\right.
		\end{align*}
		where $\mathsf{K} = \min\{K,Q\} $.
Similar result can be obtained for the optimal choice of $(r^*_2,K^*_s,K^*_h)$ if the ACDC scheme is optimal, and is skipped due   to   page limit. 
		}
	\end{remark}


	\section{General Achievable Scheme} \label{Scheme}
	In this section, we describe a hybrid scheme for the case $s\geq 1$ where each reduce function is computed by $s$ nodes.   
	
	Divide $K$ nodes into two disjoint sets $\mathcal{K}_s$ and $\mathcal{K}_h$, with $\mathcal{K}_s\subseteq\mathcal{K}$ and $\mathcal{K}_h\subset\mathcal{K}$.  Nodes in $\mathcal{K}_s$ are called ``solver'' nodes as they are responsible  for producing Reduce functions, and nodes in   $\mathcal{K}_s$ are called ``helper'' and don't  produce  any Reduce function.  		Let $K_s\triangleq|\mathcal{K}_s|$ and $K_h\triangleq|\mathcal{K}_h|$, then    $K_s+K_h\leq K$.   Let  $N_1 \triangleq \alpha N$ and $N_2 \triangleq  \bar{\alpha}N$, for some $\alpha \in [0,1]$.  Split $N$ input files into two disjoint groups $\mathcal{N}_1$ and $\mathcal{N}_2$ with  $|\mathcal{N}_1|=N_1$ and  $|\mathcal{N}_2|=N_2$. Assign them to nodes in $\mathcal{K}_s$ and  $\mathcal{K}_h$, respectively. Let $\mathcal{N}_1$   be  the assignment indices of files  on   nodes  in  $\mathcal{K}_s$, but not on any node in  $\mathcal{K}_h$, i.e.,
	\[\mathcal{N}_1=\{n: n \in \cup_{k_s\in\mathcal{K}_s}\mathcal{M}_{k_s}, n\notin \cup_{k\in\mathcal{K}_h} \mathcal{M}_k \},
	\] and 
	  $\mathcal{N}_2$ be  the assignment indices of files  on   nodes $\mathcal{K}_h$, i.e.,
 		\[\mathcal{N}_2\triangleq\cup_{k\in\mathcal{K}_h } \mathcal{M}_{k}=\mathcal{N}\backslash \mathcal{N}_1,\]
where the last equality holds because every file must be mapped by at least one  node in $\mathcal{K}_s\cup \mathcal{K}_h$.


The key idea is as follows:	 In Subsystem 1, each  file $\omega_n$ is mapped by the solver node $k_s\in\mathcal{K}_s$ if  $n\in\mathcal{M}_{k_s}\cap \mathcal{N}_1$. The corresponding mapped intermediate values are exchanged   among nodes in $\mathcal{K}_s$ during the Shuffle phase. In Subsystem 2, each  file $\omega_n$ is mapped by node $k\in\mathcal{K}_s\cup \mathcal{K}_h$ if  $n\in\mathcal{M}_{k}\cap \mathcal{N}_2$, and the resulted intermediate values are only transferred by nodes in  $\mathcal{K}_h$. After the  Shuffle phase,  the solver nodes in $\mathcal{K}_s$ reconstruct the desired intermediate values from the two subsystems and produce the assigned Reduce functions.

	
In more detail,	the Map and Shuffle processes Subsystem 1  are identical to that in the CDC scheme introduced in \cite{b1}, but with  computing nodes $\mathcal{K}_s$ and {input} files $\mathcal{N}_1$. Thus, with a peak computation load $r_1 \in \{1, \ldots, K_s\}$, the communication load of Subsystem 1 is 
	\begin{IEEEeqnarray}{rCl}
	L_1(r_1,s,K_s)  &\triangleq & \mathop{\sum}\limits_{\ell = \max\{r_1+1,s\}}^{\min\{r_1+s,K_s\}} \frac{\binom{K_s}{\ell} \binom{\ell-1}{r_1} \binom{r_1}{\ell-s}}{ \binom{K}{r_1} \binom{K_s}{s}}  \frac{\ell}{\ell - 1}.~~
\end{IEEEeqnarray}
	
	For Subsystem 2, we use a generalized ACDC scheme which extends the idea in \cite{b2} proposed for case $s=1$ to the  cascaded case ($s \ge 1$).    Denote $L_2(r_2,s,K_s)$ as the communication load of  Subsystem 2 with peak computation load  $r_2$, for some nonnegative integer $ r_2\in\{0,\ldots, K_s\}$.    
	For the trivial case $r_2 = K_s$, each file is mapped by all the $K_s$ solvers so there is no need for data shuffle, resulting in $L_2 = 0$.

	\subsubsection{Map phase}
 Firstly, divide input files $\mathcal{N}_2$ evenly into $K_h\binom{K_s}{r_2}$ disjoint 
 batches of size  $\eta_1\triangleq\frac{N_2}{K_h\binom{K_s}{r_2}}$, each corresponding to a subset $T\subset \mathcal{K}_s$ and index $k\in\mathcal{K}_h$, i.e.,
 \begin{IEEEeqnarray}{rCl}
\mathcal{N}_2=\cup_{k\in\mathcal{K}}\cup_{\mathcal{T}\subset\mathcal{K}_s:|\mathcal{T}|=r_2 }\mathcal{B}_{\mathcal{T},k},
\end{IEEEeqnarray}
 where    $\mathcal{B}_{\mathcal{T},k}$ denotes the batch of $\frac{N_2}{K_h\binom{{K}_s}{r_2}}$ files about the subset $\mathcal{T}$ and index $k$. Each solver node  $k_s\in \mathcal{K}_s$ maps files in $\mathcal{B}_{\mathcal{T},k}$ if $k_s\in {\mathcal{T}}$,  $\forall k\in\mathcal{K}_h$. 
  Each helper node $k_h\in\mathcal{K}_h$  maps    files in $\mathcal{B}_{\mathcal{T},k_h}$ for all $\mathcal{T}$. 
  
  After the Map phase, each solver node $k_s\in\mathcal{K}_s$ obtains  local intermediate values $\{v_{q,n} : q \in\mathcal{Q},  n \in \mathcal{M}_{k_s}\cap \mathcal{N}_2\}$ with $\left| \mathcal{M}_{k_s} \cap \mathcal{N}_2\right| = K_h\binom{K_s-1}{r_2-1} \eta_{1} = \frac{r_2 N_2}{K_s}$,  and each helper node $k_h\in \mathcal{K}_h$ obtains $\{v_{q,n} : q \in\mathcal{Q},  n \in \mathcal{M}_{k_h}\cap \mathcal{N}_2$ with $\left| \mathcal{M}_{k_h}\cap \mathcal{N}_2\} \right| = \binom{K_s}{r_2} \eta_{1} = \frac{N_2}{K_h}$.

	\subsubsection{Reduce phase}
	Divide $Q$ Reduce functions  evenly into $\binom{K_s}{s} $ disjoint groups, i.e., each group contains $\eta_2\triangleq\frac{Q}{\binom{K_s}{s}}$ functions and corresponds to a subset $\xi$ of size $s$. $\mathcal{R}_{\xi}$ denotes the group of Reduce functions computed exclusively by the solvers in $\xi$. Given the allocation above, for each solver $k_s \in \mathcal{K}_s$, if $k_s \in \xi$, it maps all the files in $\mathcal{R}_{\xi}$. Each solver is in $\binom{K_s-1}{s-1}$ subsets of size $s$, thus it is responsible for computing $\left| \mathcal{W}_{k_s} \right| = \binom{K_s-1}{s-1} \eta_{2} = \frac{sQ}{K_s}$ Reduce functions, for all $k_s \in \mathcal{K}_s$.

	\subsubsection{Shuffle phase}
	Only the helpers shuffle, solvers just receive messages from helpers and decode them to recover needed intermediate values
	
	For a subset $\mathcal{S} \subseteq \mathcal{K}_s$, and a subset $\mathcal{S}_1 \subset \mathcal{S} : \left| \mathcal{S}_1 \right| = r_2$, {denote the set of intermediate values required by all solvers in $\mathcal{S} \setminus \mathcal{S}_1$  while exclusively known by both the helper $k_h \in \mathcal{K}_h$ and  $r_2$ solvers in $\mathcal{S}_1$ as $V_{k_h,\mathcal{S}_1}^{\mathcal{S} \setminus \mathcal{S}_1}$,} i.e.,
	{\begin{align} \label{eqV}
		&V_{k_h,\mathcal{S}_1}^{\mathcal{S} \setminus \mathcal{S}_1} \triangleq \{v_{q,n}: q \in \mathop{\cap}\limits_{k_s \in {\mathcal{S} \setminus \mathcal{S}_1}} \mathcal{W}_{k_s}, q \notin \mathop{\cup}\limits_{k_s \notin \mathcal{S} } \mathcal{W}_{k_s}, 
		\notag
		\\& w_n \in \mathop{\cap}\limits_{k_s \in {\mathcal{S}_1}} \mathcal{M}_{k_s} \bigcap \mathcal{M}_{k_h}, w_n \notin \mathop{\cup}\limits_{k_s \notin \mathcal{S}_1} \mathcal{M}_{k_s} \bigcup \mathcal{M}_{k_h}\}.
		\end{align}}
	
	Similarly, there are $\binom{r_2}{\left| \mathcal{S} \right| - s}\eta_2$ output functions with needed intermediate values only in $V_{k_h,\mathcal{S}_1}^{\mathcal{S} \setminus \mathcal{S}_1}$. So $V_{k_h,\mathcal{S}_1}^{\mathcal{S} \setminus \mathcal{S}_1}$ contains  $\binom{r_2}{\left| \mathcal{S} \right| - s}\eta_2\eta_1$ intermediate values.
	
	\par \textit{a) Encoding:} 
	Create a symbol $ U_{k_h,\mathcal{S}_1}^{\mathcal{S}\setminus\mathcal{S}_1}\! \!\!\in\!\! \mathbb{F}_{2^{\binom{r_2}{\!\left| \mathcal{S} \right| - s}\eta_2\eta_1 T\!}}$ by concatenating all  intermediate values in $ V_{k_h,\mathcal{S}_1}^{\mathcal{S}\setminus\mathcal{S}_1}$.  Given the set $\mathcal{S}$, there are $n_1 \triangleq \binom{\left| \mathcal{S} \right|}{r_2}$ subsets each with cardinality of $r_2$. Denote these sets as $\mathcal{S}[1],\ldots,\mathcal{S}[n_1]$, and the corresponding message symbols are  $U_{k_h,\mathcal{S}[1]}^{\mathcal{S} \setminus \mathcal{S}[1]}, U_{k_h,\mathcal{S}[2]}^{\mathcal{S} \setminus \mathcal{S}[2]}, \ldots, U_{k_h,\mathcal{S}[n_1]}^{\mathcal{S} \setminus \mathcal{S}[n_1]}$. After the Map phase, the   helper $k_h$ knows all the intermediate values needed by the solvers in $(\mathcal{S}, k_h)$, so it broadcasts $n_2 \triangleq \binom{\left| \mathcal{S} \right| - 1}{r_2}$ linear combinations of the $n_1$ message symbols to the solvers in $\mathcal{S}$, denoted by $Y_{k_h}^{\mathcal{S}} \left[ 1 \right], Y_{k_h}^{\mathcal{S}} \left[ 2 \right], \ldots, Y_{k_h}^{\mathcal{S}} \left[ n_2 \right], $ for some coefficients $\alpha_1, \alpha_2, \ldots, \alpha_{n_1}$ distinct from one another and $\alpha_i \in \mathbb{F}_{2^{\binom{r_2}{\left| \mathcal{S} \right| - s}\eta_1 \eta_2 T}} $ for all $i = 1, \ldots, n_2$, i.e., 
$~~\begin{bmatrix}
	Y_{k_h}^{\mathcal{S}}[1]\\
	Y_{k_h}^{\mathcal{S}}[2]\\
	\vdots\\
	Y_{k_h}^{\mathcal{S}}[n_2]
	\end{bmatrix}
	\!=\!
	\begin{bmatrix}
	1 & 1 & \cdots & 1 \\
	\alpha_1 & \alpha_2 & \cdots & n_1 \\
	\vdots & \vdots & \ddots & \vdots \\
	\alpha_1^{n_2-1} & \alpha_2^{n_2-1} & \cdots & \alpha_{n_1}^{n_2-1} \\
	\end{bmatrix}\!\!
	\begin{bmatrix}
	U_{k_h,\mathcal{S}[1]}^{\mathcal{S} \setminus \mathcal{S}[1]}\\
	U_{k_h,\mathcal{S}[2]}^{\mathcal{S} \setminus \mathcal{S}[2]}\\
	\vdots\\
	U_{k_h,\mathcal{S}[n_1]}^{\mathcal{S} \setminus \mathcal{S}[n_1]}
	\end{bmatrix}\!\!.$
	\par \textit{b) Decoding: }
	Since each solver $k_s \in \mathcal{S}$ with helper $k_h$ is in $\binom{\left| \mathcal{S} \right| - 1}{r_2 - 1}$ subsets of $\mathcal{S}$ with size $r_2$, so it knows $\binom{\left| \mathcal{S} \right| - 1}{r_2 - 1}$ of the message symbols. When node $k_{s}$  receives the messages from $k_{h}$, it removes the known segments from each $Y_{k_h}^{\mathcal{S}}[i]$, generating new message $Z_{k_h}^{\mathcal{S}}[i]$ with only $n_1 - \binom{\left| \mathcal{S} \right| - 1}{r_2 - 1} = \binom{\left| \mathcal{S} \right|}{r_2} - \binom{\left| \mathcal{S} \right| - 1}{r_2 - 1} = \binom{\left| \mathcal{S} \right| - 1}{r_2} = n_2$ message symbols. So, there are $n_{2}$ new messages and an invertible Vandermonde matrix which is a submatrix of the encoding matrix above. As a result, the node $k_{s}$  {can} decode the rest $n_{2}$ message symbols, obtaining all  intermediate  {values} needed from $(\mathcal{S}, k_{h})$. 
	
	In the Shuffle phase, for each subset $\mathcal{S} \subseteq\mathcal{K}_s$ of size $\max\{r_2+1, s\} \leq \left| \mathcal{S} \right| \leq \min\{r_2+s, K_s\}$, each helper node multicasts $ n_2 = \binom{\left| \mathcal{S} \right| - 1}{r_2}$ message symbols to the solvers in $\mathcal{S}$, each message symbol containing $\binom{r_2}{\left| \mathcal{S} \right| - s} T $ bits, and there are $K_{h}$ helpers multicasting such message symbols. Therefore, the communication load of Subsystem 2    is 
	\begin{IEEEeqnarray}{rCl}
	L_2(r_2, s,K_s)  &= & \mathop{\sum}\limits_{\ell = \max\{r_2+1, s\}}^{\min\{r_2+s, K_s\}} \frac{\binom{K_s}{\ell} \binom{\ell - 1}{r_2} \binom{r_2}{\ell - s}}{\binom{K_s}{r_2} \binom{K_s}{s}}.
\end{IEEEeqnarray}
	
	{ Based on the scheme described above and according to  Definition \ref{DefTime}, we obtain the total number of stored files among $K$ nodes as
	\begin{IEEEeqnarray}{rCl}\label{Storage2}
	\sum_{k_s\in\mathcal{K}_s}|\mathcal{M}_{k_s}|+\sum_{k_h\in\mathcal{K}_h}|\mathcal{M}_{k_s}|&=&\alpha r_1N+\bar{\alpha}r_2N, 
\end{IEEEeqnarray}
and the
	Map, Shuffle and Reduce time as 
	\begin{IEEEeqnarray}{rCl}
	T_\text{map}&=&c_\textnormal{m} \max \left\{ \frac{\alpha r_1 +  \bar{\alpha} r_2}{K_s} ,\frac{\bar{\alpha} }{K_h}\right\},\\
	T_\text{shuffle}&=& \alpha  c_\textnormal{s} L_1(r_1,s,K_s)  +\bar{\alpha} c_\textnormal{s}  L_2(r_2,s,K_s),\\
	T_\text{reduce}&=& c_\textnormal{r}\frac{sQ}{K_s}.
\end{IEEEeqnarray}
By storage constraint in \eqref{Storage} and \eqref{Storage2}, we have
\begin{IEEEeqnarray}{rCl}
\alpha r_1N+\bar{\alpha}r_2N\leq M. 
\end{IEEEeqnarray}
	Since we focus on the case  where the Map, Shuffle and Reduce phases proceed  in a sequential fashion, we have
	\begin{IEEEeqnarray}{rCl}
T_{\mathcal{W}}(s)&= & T_\textnormal{map} +T_\textnormal{shuffle} +T_\textnormal{reduce} \nonumber\\
		&= &  	c_\text{m} \max \left\{\alpha \frac{r_1}{K_s} \!+ \!\bar{\alpha} \frac{r_2}{K_s},\frac{1\!-\! \alpha}{K_h}\right\} +\! c_\textnormal{r}\frac{sQ}{K_s}
		\nonumber\\&& +c_\textnormal{s} \alpha L_1(r_1, s,K_s)+ c_\textnormal{s} \bar{\alpha}   L_2(r_2, s,K_s) ,
\end{IEEEeqnarray}
which completes the achievability proof of Theorem \ref{theorem T}. }

	\section{Converse proof}\label{converse}
	
	\subsection{Lower bound of communication load}
	We first introduce a lemma presented in \cite{b2}.
	\begin{Lemma}\label{Lemma1}
	Consider a distributed computing task with  $N$ input files, $Q$ Reduce functions, a  file placement $\{\mathcal{M}_k\}_{k=1}^K$ and   a Reduce design  $\{\mathcal{W}_k\}_{k=1}^K$ that  uses  $K$  nodes.
	Let $a_{j,d}$ denote the number of intermediate values that are available at $j$ nodes and required by (but not available at) $d$ nodes. The following lower bound on the peak communication load holds:
\begin{IEEEeqnarray}{rCl}
	L^* \ge \frac{1}{QN} \sum_{j= 1}^{K}\sum_{d = 1}^{K - j} a_{j,d}\frac{d}{j + d - 1}.
\end{IEEEeqnarray}

	\end{Lemma}

	{{ For any scheme with a Reduce design $\{\mathcal{W}_{k}\}_{k=1}^{K}$, each  node either produces Reduce functions or not. Thus, the nodes can  be characterized into two categories: $\mathcal{K}_h$ containing nodes who will not perform any Reduce function   and $\mathcal{K}_s$  containing the remaining   nodes, i.e.,  
	\begin{IEEEeqnarray}{rCl}
	\mathcal{K}_h \triangleq\{k:\mathcal{W}_k=\emptyset,k\in\mathcal{K}\},  \mathcal{K}_s \triangleq \mathcal{K} \backslash \mathcal{K}_h.
\end{IEEEeqnarray}
 Let  $K_h\triangleq| \mathcal{K}_h|$ and $K_s\triangleq |\mathcal{K}_s|$. Note that   $\mathcal{K}_h$ and $\mathcal{K}_s$ are fixed  once the Reduce design  $\{\mathcal{W}_{k}\}_{k=1}^{K}$ is given, independent of the Map and Shuffle operations.   

 For any scheme with a file placement $\{\mathcal{M}_{k}\}_{k=1}^K$, 	let $\mathcal{N}_1$   be  the assignment indices of files  on   nodes  in  $\mathcal{K}_s$, but not on any node in  $\mathcal{K}_h$, i.e.,
	\[\mathcal{N}_1=\{n: n \in \cup_{k_s\in\mathcal{K}_s}\mathcal{M}_{k_s}, n\notin \cup_{k\in\mathcal{K}_h} \mathcal{M}_k \},
	\] and 
	  $\mathcal{N}_2$ be  all assignment indices of files  on   nodes in $\mathcal{K}_h$, i.e.,
		\[\mathcal{N}_2\triangleq\cup_{k\in\mathcal{K}_h } \mathcal{M}_{k}.\]
Every file must be mapped by at least one node in $\mathcal{K}=\mathcal{K}_s\cup \mathcal{K}_h$, indicating that  $\mathcal{N}_2=\mathcal{N}\backslash \mathcal{N}_1$. 

   Let $b_{j,1}$ be the number of files  which are stored at $j$ nodes in $\mathcal{K}_s$, but  not   at any node in $\mathcal{K}_h$, then we have $|\mathcal{N}_1| =\sum_{j = 0}^{K_s}b_{j,1}$.  Let   $b_{j,2}$ be the number of files stored at  $j$ nodes in $\mathcal{K}_s$ and at least one node in  $\mathcal{K}_h$ at the same time, then  we have   
	$|\mathcal{N}_2| = \sum_{j = 0}^{K_s}b_{j,2}.$  	Let
	\begin{align}
	&\alpha \triangleq \mathop{\sum}\limits_{j=0}^{K} \frac{b_{j,1}}{N},  ~ r_1 \triangleq \mathop{\sum}\limits_{j=0}^{K} \frac{jb_{j,1}}{\alpha N},   ~ r_2 \triangleq \mathop{\sum}\limits_{j=0}^{K} \frac{jb_{j,2}}{\bar{\alpha} N}. \label{conv4}
	\end{align}
	Since $|\mathcal{N}_1| + |\mathcal{N}_2| = N$, we have $\bar{\alpha}= \frac{|\mathcal{N}_2|}{N}= \mathop{\sum}_{j=0}^{K} \frac{b_{j,2}}{N}$.  By storage {constraint} in \eqref{Storage} and \eqref{conv4}, we have 
	\begin{IEEEeqnarray*}{rCl}
	 \sum_{k \in\mathcal{K}} |\mathcal{M}_{k}|=\mathop{\sum}\limits_{j=0}^{K} {jb_{j,1}}+\mathop{\sum}\limits_{j=0}^{K} {jb_{j,2}}=\alpha r_1N+\bar{\alpha}r_2N\leq M.
\end{IEEEeqnarray*}

	   Similar to \cite{b2}, we introduce an enhanced distributed computing system by merging  all nodes in $ \mathcal{K}_h$ into a super node such that all  files in $\mathcal{N}_2$ can be evenly mapped by nodes in $ \mathcal{K}_h$ in parallel,  and the mapped intermediate values can  be shared without data shuffle.

 For this enhanced system,	let $a^1_{j,d}$ be the number of intermediate values that are known by $j$ nodes in $\mathcal{K}_s$, not mapped by the super  node, and needed by (but not available at) $d$ nodes;  $a^2_{j,d}$ be the number of intermediate values that are mapped both by $j$ nodes in $\mathcal{K}_s$ and the super node, and needed by (but not available at) $d$ nodes in $\mathcal{K}_s$. 
	
	According to Lemma \ref{Lemma1}, we have 
	\begin{IEEEeqnarray}{rCl}
		L^* &\geq &   \frac{1}{QN} \sum_{j= 1}^{K}\sum_{d = 1}^{K - j} a_{j,d}\frac{d}{j + d - 1}\nonumber\\
		&\stackrel{(a)}{=}& {\sum_{j=1}^{K_s} \sum\limits_{d=1}^{K_s-j} \frac{a^1_{j,d}}{QN}  \frac{d}{j+d-1}}\!+\!{\sum_{j=1}^{K_s+1} \sum_{d=1}^{K_s+1-j} \frac{a^2_{j,d}}{QN}  \frac{d}{j+d-1}}\nonumber\\ 
		&=& \frac{1}{QN} \mathop{\sum}\limits_{j=0}^{K_s} \mathop{\sum}\limits_{d=\max\{1,s\!-\!j\}}^{\min\{s,K_s\!-\!j\}} \left(a^1_{j,d} \frac{d}{j\!+\!d\!-\!1} \!+\! a^2_{j,d} \frac{d}{j+d}\right)\nonumber\\
		&=&\alpha L_1^* +\bar{\alpha}  L_2^*\label{eqLL1L2},
	\end{IEEEeqnarray}
		where     equality (a) holds by  definitions of $a^1_{j,d}$ and  $a^2_{j,d}$, and 
		\begin{IEEEeqnarray}{rCl}
		 L_1^*&\triangleq& \frac{1}{\alpha NQ} \mathop{\sum}\limits_{j=0}^{K_s} \mathop{\sum}\limits_{d=\max\{1,s-j\}}^{\min\{s,K_s-j\}} a^1_{j,d} \frac{d}{j\!+\!d\!-\!1}, \label{eqL1L1}\\
		  L_2^*&\triangleq& \frac{1}{\bar{\alpha} NQ} \mathop{\sum}\limits_{j=0}^{K_s} \mathop{\sum}\limits_{d=\max\{1,s-j\}}^{\min\{s,K_s-j\}} a^2_{j,d} \frac{d}{j\!+\!d}\label{eqL2L2}.
\end{IEEEeqnarray}

\begin{remark}
The lower bound in \eqref{eqLL1L2} is valid  for  all  schemes with any file placement $\{\mathcal{M}_{k}\}_{k=1}^{K}$ and Reduce design $\{\mathcal{W}_{k}\}_{k=1}^{K}$.  Although  \eqref{eqLL1L2} has a form similar to the  time sharing scheme, the parameter $\alpha$   is defined in \eqref{conv4} and can not be changed once $\{\mathcal{W}_{k}\}_{k=1}^{K}$ and $\{\mathcal{M}_{k}\}_{k=1}^{K}$ are given. 
 \end{remark}

Now we compute $a^1_{j,d}$ and  $a^2_{j,d}$ when using the weakly symmetric reduce assignment $\mathcal{W}_\textnormal{symmetric}$  described  in Definition \ref{DefFilePm}.

	First consider the simple case where $s=1$.	In this case, each Reduce function is  mapped by only one node in  $\mathcal{K}_s$. Since each node  $k\in\mathcal{K}_s$ requires the intermediate values $v_{q,n}$, for all $q\in\mathcal{W}_k, n\notin \mathcal{M}_k$, and $|\mathcal{W}_i|=\frac{Q}{{K_s}}, \mathcal{W}_i\cap \mathcal{W}_j=\emptyset$, for all $i,j\in\mathcal{K}_s$,  we have  
	\begin{IEEEeqnarray}{rCl}\label{eqn1n2s}
	a^1_{j,d}=\frac{Q}{{K_s}}b_{j,1} ({K_s-j}), \quad a^2_{j,d}=\frac{Q}{{K_s}}b_{j,2} ({K_s-j}).
\end{IEEEeqnarray}
	
	For the general case  $s\geq 1$, recall that the Reduce design $\mathcal{W}_\textnormal{symmetric}$ assigns  all $Q$ Reduce functions  symmetrically  to nodes in $\mathcal{K}_s$, and each node $k\in\mathcal{K}_s$ computes the  Reduce functions whose indices are in the batch  $\mathcal{D}_{\mathcal{P}}$  if $k\in \mathcal{P}\subseteq\mathcal{K}_s$. Consider a file $w_n\in\mathcal{N}_1$ that is exclusively known by $j$ nodes in $\mathcal{K}_s$, and   denote these $j$ nodes as $\mathcal{K}_{s,[j]}$. Since  nodes in $\mathcal{K}\backslash \mathcal{K}_{s,[j]}$ don't access file $w_n$, there are in total $\binom{K-j}{d}$ groups of nodes of  size $d$, and each group requires $\binom{j}{s-d}\frac{Q}{\binom{K_s}{s}}=\binom{j}{j+d-s}\frac{Q}{\binom{K_s}{s}}$  intermediate values generated by $\omega_n$. Because  there are in total $b_{j,1}$ numbers of such file $w_n$, we have 
	\begin{subequations}\label{eqn1n2}
	\begin{IEEEeqnarray}{rCl}
	a^1_{j,d}&=&\frac{Q}{\binom{K_s}{s}}b_{j,1} \binom{K_s-j}{d} \binom{j}{j+d-s}.
\end{IEEEeqnarray}
With a similar analysis, we have
\begin{IEEEeqnarray}{rCl}
	a^2_{j,d}&=&\frac{Q}{\binom{K_s}{s}}b_{j,2} \binom{K_s-j}{d} \binom{j}{j+d-s}.
\end{IEEEeqnarray}
\end{subequations}		
   
   In view of the fact that    \eqref{eqn1n2} is consistent with \eqref{eqn1n2s} when $s$ equals 1,  we derive the lower bound on the  execution time  $T^*_{\mathcal{W}}(s)$ based on (\ref{conv4}--\ref{eqL2L2}) and \eqref{eqn1n2}. 

		  Let $\ell=j+d$ and substitute  \eqref{eqn1n2} into   \eqref{eqL1L1} and  \eqref{eqL2L2},  we have 
	\begin{IEEEeqnarray}{rCl}
		 L_1^*
		&=& \frac{1}{\binom{K}{s}} \mathop{\sum}\limits_{j=0}^{K} \mathop{\sum}\limits_{\ell=\max\{j+1,s\}}^{\min\{j+s,K\}} \frac{b_{j,1}}{\alpha N} \binom{K-j}{\ell-j} \binom{j}{\ell-s} \frac{\ell-j}{\ell-1}, ~\nonumber\\ \label{eqL1}
		L_2^*&=& 
		\mathop{\sum}\limits_{j=0}^{K} \mathop{\sum}\limits_{\ell=\max\{j+1,s\}}^{\min\{j+s,K\}} \frac{b_{j,2}}{\binom{K}{s}\bar{\alpha} N} \binom{K-j}{\ell-j} \binom{j}{\ell\!-\!s} \frac{\ell\!-\!j}{\ell}.  \label{eqL2}
			\end{IEEEeqnarray}
The lower bounds of $ L_1^*$ and $ L_2^*$ are illustrated  as follows. 

	\subsubsection{The Lower Bound of $L^*_1$}
	Since the $\binom{K_s\!-\!j}{\ell\!-\!j} \binom{j}{\ell-s} \frac{\ell-j}{\ell-1}$ is convex with respect to $j$, by Jensen's inequality, we have
	\begin{IEEEeqnarray}{rCl}\label{eqL1} 
		L^*_1   & {\geq} &\mathop{\sum}\limits\limits_{\ell=\max\{r_1+1,s\}}^{\min\{r_1+s,K_s\}}  \!\!\!\! \binom{K_s-\mathop{\sum}\limits_{j=0}^{K_s}{\frac{jb_{j,1}}{\alpha N}}}{\ell-\mathop{\sum}\limits_{j=0}^{K_s}{\frac{jb_{j,1}}{\alpha N}}}  \binom{\mathop{\sum}\limits_{j=0}^{K_s}{\frac{jb_{j,1}}{\alpha N}}}{\ell-s} \frac{\ell-\mathop{\sum}\limits_{j=0}^{K_s}{\frac{jb_{j,1}}{\alpha N}}}{\binom{K_s}{s}(\ell-1)} \nonumber\\
		&\stackrel{(a)}{=}  &   \mathop{\sum}\limits_{\ell=\max\{r_1+1,s\}}^{\min\{r_1+s,K_s\}} \frac{\binom{K_s}{\ell} \binom{\ell-1}{r_1} \binom{r_1}{\ell-s}\ell}{\binom{K_s}{r_1} \binom{K_s}{s}({\ell - 1} )}\stackrel{(b)}{=}       L_1(r_1, s,K_s), \label{eqLbL1}
	\end{IEEEeqnarray} 
where (a) holds by the definition of $r_1$ in \eqref{conv4}, and (b) holds by the definition of 	$L_1(r_2, s,K_s)$ in \eqref{LoadCDC}.  For the general case $0 \leq r_1 \leq K_s$, using {the} same method as in \cite{b1}, we can prove that   $L_1^*$ is lower bounded by  the lower convex envelope of the points $\{(r_1, L_1(r_1, s,K_s)): r_1 \in \{0, \ldots, K_s\} \}$.
	
	\subsubsection{The Lower Bound of $L^*_2$} 
	Since $\binom{K_s\!-\!j}{\ell\!-\!j} \binom{j}{\ell-s} \frac{\ell-j}{\ell}$ is convex with respect to $j$, by Jensen's inequality, we have
	 \begin{align}
		L^*_2 & \ge \frac{1}{\binom{K_s}{s}} \mathop{\sum}\limits_{\ell=\max\Big\{\mathop{\sum}\limits_{j=0}^{K_s}{\frac{jb_{j,2}}{\bar{\alpha} N}}+1,s\Big\}}^{\min\Big\{\mathop{\sum}\limits_{j=0}^{K_s}{\frac{jb_{j,2}}{\bar{\alpha} N}}+s,K_s\Big\}} \binom{K_s-\mathop{\sum}\limits_{j=0}^{K_s}{\frac{jb_{j,2}}{\bar{\alpha} N}}}{\ell-\mathop{\sum}\limits_{j=0}^{K_s}{\frac{jb_{j,2}}{\bar{\alpha} \!\!N}}} \nonumber
		\\& \hspace{15ex} \ \cdot \binom{\mathop{\sum}\limits_{j=0}^{K_s}{\frac{jb_{j,2}}{\bar{\alpha} N}}}{\ell-s} \frac{\ell-\mathop{\sum}\limits_{j=0}^{K_s}{\frac{jb_{j,2}}{\bar{\alpha} N}}}{\ell} \nonumber
		\\&\stackrel{(a)}{=} \mathop{\sum}\limits_{\ell=\max\{r_2+1,s\}}^{\min\{r_2+s,K_s\}}\!\! \frac{\binom{K_s}{\ell} \binom{\ell-1}{r_2} \binom{r_2}{\ell-s}}{\binom{K_s}{r_2} \binom{K_s}{s}} \stackrel{(b)}{=} L_2(r_2, s,K_s),  \label{eqLbL2}
		\end{align}
{where (a) holds by the definition of $r_2$ in \eqref{conv4}, and (b) follows by the definition of 	$L_2(r_2, s,K_s)$ in \eqref{LoadACDC}. 
	Using the same method as in \cite{b1}, we can prove that   $L_2^*$ is lower bounded by  the lower convex envelope of the points $\{(r_2,L_2(r_2, s,K_s)): r_2 \in \{0, \ldots, K_s\} \}$}.
	
	From \eqref{eqLL1L2}, \eqref{eqLbL1} and \eqref{eqLbL2}, we obtain  
	\begin{IEEEeqnarray}{rCl}\label{eqLo}
	L^*\geq \alpha   L_1(r_1, s,K_s)+ \bar{\alpha}   L_2(r_2, s,K_s).
\end{IEEEeqnarray}
	From \eqref{eqLo} and Definition \ref{DefTime},  the optimal Shuffle time, denoted by $T^*_\textnormal{shuffle}$, is  lower bounded by 
	\begin{IEEEeqnarray}{rCl}\label{ShuffleTime}
	T^*_\textnormal{shuffle} &\geq&  c_\textnormal{s} L^*\nonumber\\
	&\geq &c_\textnormal{s} \alpha L_1(r_1, s,K_s)+ c_\textnormal{s} \bar{\alpha}   L_2(r_2, s,K_s).  
\end{IEEEeqnarray}

	\subsection{Lower bounds of Map, Reduce and Execution time }
	\subsubsection{Map Time}
	From Definition \ref{DefTime},  the optimal Map  time, denoted by $T^*_\textnormal{map}$, is  lower bounded by 
	\begin{IEEEeqnarray}{rCl}\label{eqMpTime}
	T^*_\textnormal{map} & \geq & c_\text{m}p\nonumber\\
			  &\geq&  c_\text{m}\cdot \max \left\{ \frac{\sum_{k\in\mathcal{K}_s} |\mathcal{M}_k|}{K_sN}, \frac{\sum_{k\in\mathcal{K}_h} |\mathcal{M}_k|}{K_hN} \right\}\nonumber\\
			&{=}& c_\text{m}\cdot \max \left\{ \frac{\sum_{j=0}^{{K}_s} jb_{j,1}+\sum_{j=0}^{{K}_s} jb_{j,2}}{K_sN},\frac{\sum_{j = 0}^{K_s}b_{j,2}}{K_hN}\right\} \nonumber\\
			&=& c_\text{m}\cdot \max \left\{\alpha \frac{r_1}{K_s} + \bar{\alpha} \frac{r_2}{K_s},\frac{\bar{\alpha} }{K_h}\right\},
	\end{IEEEeqnarray}
	where $p$ is the peak computation load defined in Definition \ref{defCpLoad} and the last equality holds by \eqref{conv4}.
	
	\subsubsection{Reduce time} From Definition \ref{DefTime},  the optimal Reduce  time, denoted by $T^*_\textnormal{reduce}$, is  lower bounded by 
	\begin{IEEEeqnarray}{rCl}\label{eqRdTime}
		T^*_\textnormal{reduce} \geq  
		c_\textnormal{r} \cdot \max \limits_{k_s \in \{1, \dots, K_s\}} |\mathcal{W}_k| = c_\textnormal{r}\frac{sQ}{K_s}, 
	\end{IEEEeqnarray}
	where the last equality holds by \eqref{eqReduce}.
	\subsubsection{Execution time} From   (\ref{ShuffleTime}--\ref{eqRdTime}) and  Definition \ref{DefTime}, we obtain 
\begin{IEEEeqnarray}{rCl}
T^*_{\mathcal{W}}(s)
		&\geq &  	c_\text{m} \max \left\{\alpha \frac{r_1}{K_s} \!+ \!\bar{\alpha} \frac{r_2}{K_s},\frac{\bar{\alpha} }{K_h}\right\} +\! c_\textnormal{r}\frac{sQ}{K_s}
		\nonumber\\&& +c_\textnormal{s} \alpha L_1(r_1, s,K_s)+ c_\textnormal{s} \bar{\alpha}   L_2(r_2, s,K_s) ,
\end{IEEEeqnarray}		 
which completes the converse proof of Theorem \ref{theorem T}.
		}}

\end{document}